\documentclass[10pt]{article}

\usepackage{amsmath}
\usepackage{amsfonts}
\usepackage{amsmath,amssymb}
\usepackage{amssymb, amscd, amsthm}
\usepackage{graphicx}
\usepackage{color}
\usepackage{url}
\usepackage{longtable}
\usepackage{booktabs}
\usepackage{float}
\restylefloat{table}

\newcommand{\tabincell}[2]{\begin{tabular}{@{}#1@{}}#2\end{tabular}}

\newtheorem{theorem}{Theorem}
\newtheorem{example}{Example}
\newtheorem{lemma}{Lemma}
\newtheorem{definition}{Definition}

\begin{document}
\title{Differential Spectrum and Boomerang Spectrum of Some Power Mapping}

\author{Yuehui Cui,\mbox{ } Jinquan Luo}
\date{School of Mathematics and Statistics,
 Hubei Key Laboratory of Mathematical Sciences, 
 Central China Normal University,  Wuhan 430079, China}
\maketitle
\insert\footins{\small{\it Email addresses} :luojinquan@mail.ccnu.edu.cn (Jinquan Luo); hfcyh1@163.com (Yuehui Cui).}
{\centering\section*{Abstract}}
 \addcontentsline{toc}{section}{\protect Abstract}
 \setcounter{equation}{0}

Let $f(x)=x^{s(p^m-1)}$ be a power mapping over $\mathbb{F}_{p^n}$, where $n=2m$ and $\gcd(s,p^m+1)=t$. In \cite{kpm-1}, Hu et al. determined the differential spectrum and boomerang spectrum of the power function $f$, where $t=1$. So what happens if $t\geq1$? In this paper, we extend the result of \cite{kpm-1} from $t=1$ to general case. We use a different method than in \cite{kpm-1} to determine the differential spectrum and boomerang spectrum of $f$ by studying the number of rational points on some curves. 
%It is shown that the differential spectrum $f$ is $\mathbb{DS}=\left\{\omega_0,\cdots, \omega_{2t^2}, \omega_{t(p^m-1)}=1 \right\}$.
This method may be helpful for calculating the differential spectrum and boomerang spectrum of some Niho type power functions.

\medskip
\noindent{\large\bf Keywords: }\medskip Differential spectrum, Boomerang spectrum, Power mapping, Locally-APN function.

\noindent{\bf2010 Mathematics Subject Classification}: 94A60, 11T06.

\section{Introduction}

Substitution boxes (S-boxes) are vital elements in the design of secure cryptographic primitives like block ciphers. Differential attack, introduced by Biham and Shamir \cite{diffatt}, is a key cryptanalytic technique used to evaluate the security of block ciphers. To measure the level of security of a substitution box used in a block cipher against differential attacks, Nyberg \cite{diffuniform} introduced the concept of differential uniformity. Let $\mathbb{F}_{p^n}$ be the finite field with $p^n$ elements, where $p$ is a prime and $n$ is a positive integer. Let $\mathbb{F}_{p^n}^*$ denote the multiplicative group of the finite field $\mathbb{F}_{p^n}$. For any function $f:\mathbb{F}_{p^n}\rightarrow\mathbb{F}_{p^n}$, the derivative function of $f$ in respect to any element $a \in \mathbb{F}_{p^n}$  is a function from $\mathbb{F}_{p^n}$ to $\mathbb{F}_{p^n}$  defined by
\begin{equation*}
  D_af(x)=f(x+a)-f(x), \forall x \in \mathbb{F}_{p^n}.
\end{equation*}
For any $a,b \in \mathbb{F}_{p^n}$, let 
\begin{equation*}
  \delta(a,b)=\#\left\{x \in \mathbb{F}_{p^n} \mid D_af(x)=b \right\}.
\end{equation*}
The differential uniformity of $f$ is defined as
\begin{equation*}
  \delta=\max\left\{\delta(a,b) \mid a \in \mathbb{F}_{p^n}^*, b \in \mathbb{F}_{p^n}  \right\}.
\end{equation*}
If $\delta=1$, then $f$ is called a perfect nonlinear (PN) function. If $\delta=2$, then $f$ is called an almost perfect nonlinear (APN) function. Power functions are often favored as S-boxes candidates due to their simple algebraic structure and lower hardware implementation cost. Their specific algebraic properties also make it easier to analyze their differential properties. When $f(x)=x^d$ is a power mapping,
\begin{equation*}
  (x+a)^d-x^d=b \Leftrightarrow a^d \left (\left(\frac{x}{a}+1\right)^d-\left(\frac{x}{a}\right)^d \right)=b
\end{equation*}
implying that $\delta(a,b)=\delta(1,\frac{b}{a^d})$ for all $a \in \mathbb{F}_{p^n}^*$ and $b \in \mathbb{F}_{p^n}$. Thus, the differential properties of $f(x)=x^d$ are completely determined by the values of $\delta(1,b)$ when $b$ runs through $\mathbb{F}_{p^n}$. So we can have the following definition \cite{firstdef}.

\begin{definition}\rm
  Let $f(x)=x^d$ be a power function from $\mathbb{F}_{p^n}$ to $\mathbb{F}_{p^n}$ with differential uniformity $\delta$. Denote 
\begin{equation*}
  \omega_i=\#\left\{b \in \mathbb{F}_{p^n} \mid \delta(1,b)=i \right\}
\end{equation*}
for $0\leq i\leq\delta$. The differential spectrum of $f(x)=x^d$ is defined as
\begin{equation*}
  \mathbb{DS}=\left\{\omega_i \mid 0\leq i\leq\delta \mbox{ and } \omega_i>0 \right\}.
\end{equation*}\hfill$\square$
\end{definition}
A function $f$ is said to be locally-APN if $\max\left\{\delta(1,b) \mid  b \in \mathbb{F}_{p^n}\setminus\mathbb{F}_{p}  \right\}=2$. The readers can refer to \cite{APN1, locally, APN2, kpm-1} for more information on PN, APN and locally-APN functions. To our knowledge, there are only a limited number of locally-APN functions in the literature. We also have two basic identities \cite{firstdef}.
\begin{equation}\label{two basic ide}
  \sum\limits_{i=0}^{\delta}\omega_i=p^n \mbox{ and } \sum\limits_{i=0}^{\delta}i\omega_i=p^n.
\end{equation}
These two identities are useful in calculating the differential spectrum of $f$. However, finding the differential spectrum of a power mapping remains a theoretical challenge. For known results, the reader can refer to recent (survey) papers \cite{dssurvey} and \cite{dssurvey2}. The differential spectrum of a power mapping has several applications, including cross correlation distribution, Walsh transformation, $t$-designs, and coding theory. For more detailed information, readers can refer to \cite{APN1,APN2,walshcc,appcode2, dsandcc,tdesigns}.

To assess the capability of a cryptosystem in handling the boomerang attack \cite{wagnerda}, a new concept known as boomerang uniformity was introduced by Cid et al. \cite{cidc}. Subsequently, due to the complexity involved in calculating boomerang uniformity based on the original definition, Li et al. \cite{liqu} proposed an equivalent definition that can be extended to any finite field.
\begin{definition}\rm
For $a,b\in\mathbb{F}_{p^n}^*$, let $\beta_{f}(a,b)$ denote the number of solutions in $\mathbb{F}_{p^n}^2$ of the following system of equations
\begin{eqnarray}\label{boom two eq}
    \left \{
\begin{array}{lll}
f(x)-f(y)=b,\\
f(x+a)-f(y+a)=b.\\
\end{array}\right.
\end{eqnarray}
Then the boomerang uniformity of a function $f$ over $\mathbb{F}_{p^n}$ is defined as
\begin{equation*}
  \beta=\max_{a\in\mathbb{F}_{p^n}^*}\max_{b\in\mathbb{F}_{p^n}^*}\beta_{f}(a,b).
\end{equation*}
The boomerang spectrum of a function $f$ over $\mathbb{F}_{p^n}$ is defined as 
\begin{equation*}
  \{\beta_{f}(a,b) : a\in\mathbb{F}_{p^n}^*, b\in\mathbb{F}_{p^n}^*\}.
\end{equation*}
\hfill$\square$
\end{definition}
Similar to the differential spectrum, the boomerang spectrum of a power mapping $f(x)$ over $\mathbb{F}_{p^n}$ with $\beta(f)=\beta$ can be simply defined as 
\begin{equation*}
  \mathbb{BS}=\left\{\nu_i \mid 0\leq i\leq\beta  \mbox{ and } \nu_i>0 \right\},
\end{equation*}
where $\nu_i=\#\left\{b\in\mathbb{F}_{p^n}^* \mid  \beta_f(1,b)=i \right\}$. So far, only some boomerang spectrum of power functions are known, the readers can refer to \cite{edda,quliboom,yanboom} for more information. More significantly, the boomerang spectrums of S-boxes have been demonstrated to be crucial for evaluating their resistance to various forms of boomerang cryptanalysis.

The rest of this paper is organized as follows. In Section \ref{Preliminaries}, we introduce some notations and auxiliary tools. In Section \ref{luo the ds s}, we calculate the differential spectrum of $f(x)=x^{s(p^m-1)}$, where $\gcd(s,p^m+1)\geq1$, $n=2m$. In Section \ref{luo the bs}, we compute the boomerang spectrum of $f$. The conclusive remarks are given in Section \ref{conclusion}. 

For convenience, we fix the following notations in this paper:

\begin{tabular}{ll}
$\mathbb{F}_{p^n}$ & finite field of $p^n$ elements \\
$\mathbb{F}_{p^n}^*$ & multiplicative group of $\mathbb{F}_{p^n}$ \\
$\psi$ & primitive element of $\mathbb{F}_{p^n}$ \\
\end{tabular}

\section{Preliminaries}\label{Preliminaries} 
Before stating the main results, we introduce some notations and auxiliary tools. 
If $x\in\mathbb{F}_{p^n}^*$ and $x=\psi^t$ with $0\leq t<p^n-1$, we denote $ind_{\psi}(x)=t$. Let $\mathbb{F}_{p^n}^{\sharp}=\mathbb{F}_{p^n}\setminus \{0,-1\}$. We denote
\begin{equation}\label{c12}
  C_{j_1,j_2}=\{x\in\mathbb{F}_{p^n}^{\sharp}: ind_{\psi}(x+1)\equiv j_1\,\,({\rm mod}\,\,p^m+1), ind_{\psi}(x)\equiv j_2\,\,({\rm mod}\,\,p^m+1)\},
\end{equation}
where $0\leq j_1,j_2\leq p^m$, $n=2m$. Obviously,
\begin{equation*}
 \bigsqcup\limits_{0\leq j_1,j_2\leq p^m} C_{j_1,j_2}=\mathbb{F}_{p^n}^{\sharp}
\end{equation*}
(here $\bigsqcup$ indicates disjoint union).
  
\begin{lemma}\label{luo wang}
\cite[Theorem 1]{luo111}
Let $p$ be a prime number, $n=2km$ with $k$ and $m$ positive integers. Let $n_1$ and $n_2$ be two positive integers, $t=\gcd(n_1,n_2)$, $lcm(n_1,n_2)\mid p^m+1$. Let $0\leq r_1\leq n_1-1$ and $0\leq r_2\leq n_2-1$. Let $N_{p^n}(\chi)$ be the number of $\mathbb{F}_{p^n}$-rational points on the affine curve 
\begin{equation*}
  \chi: \alpha x^{n_1}+\beta y^{n_2}+1=0.
\end{equation*}
 For $\alpha\in\left\{\psi^{r_1+n_1w}\mid 0\leq w\leq \frac{p^n-1}{n_1}-1 \right\}$, 
 $\beta\in\left\{\psi^{r_2+n_2w}\mid 0\leq w\leq \frac{p^n-1}{n_2}-1 \right\}$. Then we have
\begin{itemize}
	\item[(i)] if $r_1=r_2=0$, then $N_{p^n}(\chi)=p^n+(-1)^{k-1}((n_1-1)(n_2-1)+1-t)p^{\frac{n}{2}}-t+1$.
	\item[(ii)] if $r_1=0$, $r_2\neq0$, $t\nmid r_2$, then $N_{p^n}(\chi)=p^n+(-1)^k(n_1-2)p^{\frac{n}{2}}+1$.
	\item[(iii)] if $r_1\neq0$, $r_2=0$, $t\nmid r_1$, then $N_{p^n}(\chi)=p^n+(-1)^k(n_2-2)p^{\frac{n}{2}}+1$.
	\item[(iv)] if $r_1\neq0$, $r_2\neq0$ and $t\nmid r_1-r_2$, then $N_{p^n}(\chi)=p^n+(-1)^{k-1}2p^{\frac{n}{2}}+1$.
	\item[(v)] if $r_1\neq0$, $r_2\neq0$ and $t\mid r_1-r_2$, then $N_{p^n}(\chi)=p^n+(-1)^{k}(t-2)p^{\frac{n}{2}}-t+1$.
\end{itemize}
\hfill$\square$
\end{lemma}

\begin{lemma}\label{doubesame}
Let $p$ be a prime. Let $k,m$ be positive integers, $n=2km$, $\alpha=\psi^{p^m-1}$. Then the equality $\alpha^i+\alpha^j=\alpha^{i'}+\alpha^{j'}$ with $0\leq i\leq j\leq p^m$ and $0\leq i'\leq j'\leq p^m$ implies $i=i'$, $j=j'$. 
\end{lemma}
\begin{proof}
If $\alpha^i+\alpha^j=\alpha^{i'}+\alpha^{j'}$, we get $\alpha^i\alpha^j=\alpha^{i'}\alpha^{j'}$ based on $\alpha^{p^m}=\alpha^{-1}$. Let $a=\alpha^i+\alpha^j$ and $b=\alpha^i\alpha^j$, then both $\{\alpha^i,\alpha^j\}$ and $\{\alpha^{i'},\alpha^{j'}\}$ are two solutions to $x^2-ax+b=0$.
\end{proof}

\section{Differential spectrum of $f(x)=x^{s(p^m-1)}$ over $\mathbb{F}_{p^n}$}\label{luo the ds s}
In \cite{kpm-1}, Hu et al. determined the differential spectrum of the power function $x^{s(p^m-1)}$ by studying the differential equation, where $n=2m$ and $\gcd(s,p^m+1)=1$. So what happens if $\gcd(s,p^m+1)\geq1$? In this section, we investigate the function $f(x)=x^d$ over $\mathbb{F}_{p^n}$, where $d=s(p^m-1)$, $\gcd(s,p^m+1)=t$ and $n=2m$. To determine the differential spectrum  of $x^d$, we mainly study the number of solutions in $\mathbb{F}_{p^n}$ of 
\begin{equation}\label{ds eq s f2}
  (x+1)^d-x^d=b
\end{equation}
for every $b\in \mathbb{F}_{p^n}$. Denote by $\delta(b)$ the number of solutions of (\ref{ds eq s f2}) in $\mathbb{F}_{p^n}$.

\subsection{$p=2$}

In this subsection, we investigate the function $f(x)=x^d$ over $\mathbb{F}_{2^n}$.
\begin{lemma}\label{lemma ds s f2}
Let $\alpha=\psi^{2^m-1}$ and $\frac{2^m+1}{t}>3$. With the notations above, we have 
\begin{itemize}
\item[(i)]  $\delta(0)=t(2^m-1)-1$.
\item[(ii)]  $\delta(1)=2$ for $3t\nmid2^m+1$. $\delta(1)
   =2t^2+2$ for $3t\mid2^m+1$. 
\item[(iii)] $\delta(1+\alpha^{si})
   =2t(t-1)$, where $1\leq i<\frac{2^m+1}{t}$.
\item[(iv)] $\delta(\alpha^{si}+\alpha^{sj})
   =2t^2$, where $1\leq i,j<\frac{2^m+1}{t}$, $i\neq j$ and $\alpha^{si}+\alpha^{sj}\neq1$.
\item[(v)] $\delta(b)=0$ for $b\neq 1,\alpha^{si}+\alpha^{sj}$, where $0\leq i,j<\frac{2^m+1}{t}$. 
\end{itemize}
\end{lemma}
\begin{proof} 
Let $\triangle(x)=(x+1)^{d}-x^{d}$. For any $b\in\mathbb{F}_{p^n}$, if $x_1\in C_{i,j}$ is a solution of $\triangle(x)=(x+1)^{d}-x^{d}=b$, then 
\begin{equation*}
  \triangle(x_1)=\alpha^{si}-\alpha^{sj}=b.
\end{equation*}
So, if $x\in\mathbb{F}_{p^n}^{\sharp}$, from  $\bigsqcup\limits_{0\leq i,j\leq p^m} C_{i,j}=\mathbb{F}_{p^n}^{\sharp}$ , we can distinguish $(p^m+1)^2$ disjoint cases.

(i) Note that $\triangle(x)=0$ has solutions only in $\bigcup\limits_{\substack {0\leq i< \frac{2^m+1}{t}, \\0\leq j_1,j_2<t}}C_{i+\frac{(2^m+1)j_1}{t},i+\frac{(2^m+1)j_2}{t}}$.
 For $x\in C_{0,0}$, the following holds:  
\begin{equation*}
ind_{\psi}(x) \equiv ind_{\psi}(x + 1)\equiv 0\,\,({\rm mod}\,\,2^m+1),
\end{equation*}
and
\begin{equation*}
\triangle(x)=0.
\end{equation*}
Furthermore, the number of solutions to the system
\begin{equation*}
\begin{cases}  
x_1 - x_2 + 1 = 0, \\  
ind_{\psi}(x_1) =ind_{\psi}(x_2) \equiv 0\,\,({\rm mod}\,\,2^m+1),  
\end{cases}  
\end{equation*}
is equal to the number of elements in $C_{0,0}$. According to Lemma \ref{luo wang} (i), this number is exactly 
\begin{equation*}
  \frac{2^n+(2^{2m}-2^m)2^{m}-3\cdot2^m-2}{(2^m+1)^2}=2^m-2.
\end{equation*} 
For other cases, we have the following results, see Table \ref{jktable}.
\begin{table}[H]
\centering
 \caption{Number of solutions to $\triangle(x)=0$ in $C_{l_1,l_2}$}
\label{jktable}
\begin{tabular}{l l}
%\hline
\hline
%inserts double horizontal lines
$(l_1,l_2)$&
Number of solutions to $\triangle(x)=0$ (each) \\
[0.5ex]
\hline
 $(0,0)$ & $2^m-2$\\
 \hline

$(i,i)$, $0<i<(2^m+1)/t$ & $0$ \\
\hline

$\big(0,(2^m+1)j_1/t\big)$, 
$\big((2^m+1)j_1/t,0\big)$, $0< j_1<t$ & $0$ \\
\hline

\tabincell{l}{$\big((2^m+1)j_1/t,(2^m+1)j_2/t\big)$ \\
for $0< j_1,j_2<t$, $j_1\neq j_2$}   & $1$   \\
\hline

\tabincell{l}{$\big(i+(2^m+1)j_1/t,i+(2^m+1)j_2/t\big)$\\
for $0< i< (2^m+1)/t$, $0\leq j_1,j_2<t$, $j_1\neq j_2$} & $1$  \\
\hline

Other Cases &  $0$ \\
\hline
\end{tabular}
\end{table} 
In total, 
\begin{equation*}
 \delta(0)=t(2^m-1)-1.
\end{equation*}

(ii) When $3t\nmid2^m+1$, $\triangle(x)=1$ has no solution in $\mathbb{F}_{2^n}^{\sharp}$. We also obtain $\triangle(0)=1$ and $\triangle(1)=1$. So, $\delta(1)=2$.

When $3t\mid2^m+1$, $\alpha^{\frac{2^m+1}{3}}+\alpha^{\frac{2(2^m+1)}{3}}=1$. Equation $\triangle(x)=1$ has solutions only in 
\begin{equation*}
  \bigcup\limits_{0\leq j_1,j_2<t}\bigg(C_{2l+\frac{(2^m+1)j_1}{t},
  l+\frac{(2^m+1)j_2}{t}}
  \bigcup C_{l+\frac{(2^m+1)j_1}{t},
  2l+\frac{(2^m+1)j_2}{t}}\bigg)\bigcup\{0,1\},
\end{equation*}  
where $ls(2^m-1)\equiv \frac{2^n-1}{3}\,\,({\rm mod}\,\,2^n-1)$, $1\leq l<\frac{2^m+1}{t}$.
If $x\in C_{2l,l}$, then $\triangle(x)=1$ is equivalent to the system of equations
\begin{eqnarray}\label{001alpha12}
    \left \{
\begin{array}{lll}
\psi^{2l}x_1-\psi^{l} x_2=1,\\
ind_{\psi}(x_1)=ind_{\psi}(x_2)\equiv 0\,\,({\rm mod}\,\,2^m+1).\\
\end{array}\right.
\end{eqnarray}
By Lemma \ref{luo wang} (iv), we have 
\begin{equation*}
  \frac{2^n+2^{m+1}+1}{(2^m+1)^2}=1
\end{equation*}
solution to (\ref{001alpha12}). Using the same method, we can ultimately obtain $2t^2+2$ solutions for the equation $\triangle(x)=1$ in 
\begin{equation*}
  \bigcup\limits_{0\leq j_1,j_2<t}\bigg(C_{2l+\frac{(2^m+1)j_1}{t},
  l+\frac{(2^m+1)j_2}{t}}
  \bigcup C_{l+\frac{(2^m+1)j_1}{t},
  2l+\frac{(2^m+1)j_2}{t}}\bigg)\bigcup\{0,1\}.
\end{equation*}

(iii)-(iv) Note that $\triangle(x)=1+\alpha^s$ has solutions only in 
\begin{equation*}
  \bigcup\limits_{0\leq j_1,j_2<t}\bigg(C_{\frac{(2^m+1)j_1}{t},1+\frac{(2^m+1)j_2}{t}}
\bigcup C_{1+\frac{(2^m+1)j_1}{t},\frac{(2^m+1)j_2}{t}}\bigg),
\end{equation*}
and $\triangle(x)=\alpha^s+\alpha^{2s}$ has solutions only in 
\begin{equation*}
  \bigcup\limits_{0\leq j_1,j_2<t}\bigg(C_{2+\frac{(2^m+1)j_1}{t},1+\frac{(2^m+1)j_2}{t}}
\bigcup C_{1+\frac{(2^m+1)j_1}{t},2+\frac{(2^m+1)j_2}{t}}\bigg).
\end{equation*}
By Lemma \ref{luo wang}, we obtain
\begin{equation*}
 \delta(1+\alpha^s)=2t(t-1),
\end{equation*} 
and
\begin{equation*}
 \delta(\alpha^s+\alpha^{2s})
 =2t^2.
\end{equation*} 

(v) The obvious conclusion follows from (i)-(iv) and Lemma \ref{doubesame}.
\end{proof}

We determine the differential spectrum of $f$ in the following result.

\begin{theorem}
Let $f(x)=x^d$ be a power mapping defined over $\mathbb{F}_{2^n}$, where $d=s(2^m-1)$, $\gcd(s,2^m+1)=t$, $n=2m$ and $\frac{2^m+1}{t}>3$. For any $b\in\mathbb{F}_{2^n}$, the differential spectrum of $f$ is shown as following.

(i) For case $3t\nmid2^m+1$, Table \ref{2222} holds.

(ii) For case $3t\mid2^m+1$, Table \ref{ds22table} holds.
\begin{table}[H]
\centering
 \caption{$3t\nmid2^m+1$}
 \label{2222}
\begin{tabular}{l l}
%\hline
\hline
%inserts double horizontal lines
$\delta(1,b)$&  Frequency \\
[0.5ex]
\hline
0  &  $\frac{t^22^{n+1}-(2^m-t+2)2^m-4t^2+t-1}{2t^2}$  \\

$2t^2$  & $\frac{2^{2m}+(2-3t)2^m+2t^2-3t+1}{2t^2}$ \\

$2t(t-1)$  &  $\frac{2^m-t+1}{t}$\\

2  &1  \\

$t(2^m-1)-1$  &1 \\
\hline
\end{tabular}
\end{table}

\begin{table}[H]
\centering
 \caption{$3t\mid2^m+1$}
\label{ds22table}
\begin{tabular}{l l}
\hline
$\delta(1,b)$&  Frequency \\
[0.5ex]
\hline
0  &  $\frac{t^22^{n+1}-(2^m-t+2)2^m-2t^2+t-1}{2t^2}$\\

$2t^2$   &  $\frac{2^{2m}+(2-3t)2^m-3t+1}{2t^2}$  \\

$2t(t-1)$ &  $\frac{2^m-t+1}{t}$\\

$2t^2+2$  & 1  \\

$t(2^m-1)-1$  & 1\\
\hline
\end{tabular}
\end{table}
\end{theorem}
\begin{proof}
When $3t\nmid2^m+1$, from (\ref{two basic ide}) and Lemma \ref{lemma ds s f2}, we have 
\begin{eqnarray*}
  &&\omega_{2}=1,\mbox{ } \omega_{t(2^m-1)-1}=1, \\
  &&\omega_{2t(t-1)}=\frac{2^m-t+1}{t},\\ 
  &&\omega_{2t^2}
  =\frac{2^{2m}+(2-3t)2^m+2t^2-3t+1}{2t^2},\\ 
  &&\omega_0+\omega_{2}
  +\omega_{t(2^m-1)-1}+\omega_{2t(t-1)}+\omega_{2t^2}=2^n. 
\end{eqnarray*}
The conclusion can be obtained by solving the above equations.

When $3t\mid2^m+1$, from (\ref{two basic ide}) and Lemma \ref{lemma ds s f2}, we have
\begin{eqnarray*}
  &&\omega_{2t^2+2}=1,  \\
  &&\omega_{t(2^m-1)-1}=1,  \\
  &&\omega_{2t(t-1)}=\frac{2^m-t+1}{t},\\
 &&\omega_{2t^2}
  =\frac{2^{2m}+(2-3t)2^m-3t+1}{2t^2},  \\
  &&\omega_0+\omega_{2t^2+2}
  +\omega_{2t(t-1)}+\omega_{t(2^m-1)-1}
  +\omega_{2t^2}=2^n.
\end{eqnarray*}
Solve the above equations and we obtain the result.
\end{proof}

\subsection{$p=3$}

In this subsection, we determine the differential spectrum of $f(x)=x^d$ over $\mathbb{F}_{3^n}$. Following a similar approach to the proof of Lemma \ref{lemma ds s f2}, we establish analogous results here while omitting the detailed proof.
\begin{lemma}\label{lemma ds s f3} 
Let $\alpha=\psi^{3^m-1}$ and $\frac{3^m+1}{t}>3$. Using the aforementioned notations, we obtain the following result:
\begin{itemize}
\item[(i)]  $\delta(0)=t(3^m-1)-1$.
\item[(ii)]  $\delta(1)=\delta(2)=1$
     for $2t\nmid3^m+1$. $\delta(1)=\delta(2)
     =t^2-t+1$ for $2t\mid3^m+1$.
\item[(iii)] if $2t\nmid3^m+1$, then $\delta(1-\alpha^{si})
    =\delta(\alpha^{si}-1)
    =t^2-t$, where $1\leq i< \frac{3^m+1}{t}$.
\item[(iv)] if $2t\mid3^m+1$, then $\delta(1-\alpha^{si})
    =\delta(\alpha^{si}-1)
    =2t^2-t$, where $1\leq i< \frac{3^m+1}{t}$ and $\alpha^{si}\neq2$.
\item[(v)] $\delta(2\alpha^{si})
    =0$ for $2t\nmid3^m+1$ and $1\leq i< \frac{3^m+1}{t}$. $\delta(2\alpha^{si})=t^2$ for $2t\mid3^m+1$, $1\leq i< \frac{3^m+1}{t}$ and $\alpha^{si}\neq2$.
\item[(vi)] if $2t\nmid3^m+1$, then $\delta(\alpha^{si}-\alpha^{sj})
    =t^2$, where $1\leq i,j<\frac{3^m+1}{t}$ and $i\neq j$.
\item[(vii)] if $2t\mid3^m+1$, then $\delta(\alpha^{si}-\alpha^{sj})
    =2t^2$, where $1\leq i,j< \frac{3^m+1}{t}$, $\alpha^{si},\alpha^{sj}\neq2$, $i\neq j$ and $\alpha^{si}-\alpha^{sj}\neq1,2$.
\item[(viii)] $\delta(b)=0$ for $b\neq\pm1, \alpha^{si}-\alpha^{sj}$, where $0\leq i,j< \frac{3^m+1}{t}$. \hfill$\square$
\end{itemize}
\end{lemma}
By analyzing the results of Lemma \ref{lemma ds s f3}, we obtain the following results.
\begin{theorem}
Let $f(x)=x^d$ be a power mapping defined over $\mathbb{F}_{3^n}$, where $d=s(3^m-1)$, $\gcd(s,3^m+1)=t$, $n=2m$ and $\frac{3^m+1}{t}>3$. For any $b\in\mathbb{F}_{3^n}$, the differential spectrum of $f$  is shown as following.

(i) For case $2t\nmid3^m+1$, Table \ref{ds31table} holds.

(ii) For case $2t\mid3^m+1$, Table \ref{ds32table} holds.
\begin{table}[H]
\centering
 \caption{$2t\nmid 3^m+1$}
\label{ds31table}
\begin{tabular}{l l}
\hline
$\delta(1,b)$&  Frequency \\
[0.5ex]
\hline
0  & $\frac{t^23^n-(3^m+2-t)3^m-3t^2+t-1}{t^2}$ \\

$t^2$ &
$\frac{3^{2m}+(2-3t)3^m+2t^2-3t+1}{t^2}$ \\

$t^2-t$  &  $\frac{2(3^m-t+1)}{t}$  \\

1 & 2\\

$t(3^m-1)-1$  & 1\\
\hline
\end{tabular}
\end{table}
\begin{table}[H]
\centering
 \caption{$2t\mid 3^m+1$}
\label{ds32table}
\begin{tabular}{l l}
\hline
$\delta(1,b)$&  Frequency \\
[0.5ex]
\hline
0  &  $\frac{2t^23^n-(3^m+2)3^{m}-2t^2-1}{2t^2}$ \\

$2t^2$ &  $\frac{3^{2m}+(2-6t)3^m+8t^2-6t+1}{2t^2}$ \\

$2t^2-t$  &
$\frac{2(3^m-2t+1)}{t}$  \\

$t^2$ & 
$\frac{3^m-2t+1}{t}$ \\

$t^2-t+1$ & 2\\

$t(3^m-1)-1$ & 1 \\
\hline
\end{tabular}
\end{table}
\hfill$\square$
\end{theorem}

\subsection{$p>3$}
In this subsection, we compute the differential spectrum of $x^d$ over $\mathbb{F}_{p^n}$ (for $p>3$) by employing methodology analogous to Lemma \ref{lemma ds s f2}, omitting repetitive proof details.
\begin{lemma}\label{lemma ds s f4}
Let $\alpha=\psi^{p^m-1}$ and $\frac{p^m+1}{t}>3$. Within the established notational framework, we derive:
\begin{itemize}
\item[(i)] $\delta(0)=t(p^m-1)-1$. 
\item[(ii)] $\delta(1)=\delta(p-1)=1$ for $6t\nmid p^m+1$. $\delta(1)=\delta(p-1)
     =2t^2+1$ for $6t\mid p^m+1$.
\item[(iii)]  $\delta(2)=\delta(p-2)=0$
     for $2t\nmid p^m+1$. $\delta(2)=\delta(p-2)
     =t^2-t$ for $2t\mid p^m+1$.
\item[(iv)] if $2t\nmid p^m+1$, then $\delta(1-\alpha^{si})
    =\delta(\alpha^{si}-1)
    =t^2-t$, where $1\leq i< \frac{p^m+1}{t}$.
\item[(v)] if $2t\mid p^m+1$, then $\delta(1-\alpha^{si})
    =\delta(\alpha^{si}-1)
    =2t^2-t$, where $1\leq i< \frac{p^m+1}{t}$ and $\alpha^{si}\neq-1$.
\item[(vi)] $\delta(2\alpha^{si})
    =0$ for $2t\nmid p^m+1$ and $1\leq i< \frac{p^m+1}{t}$.
$\delta(2\alpha^{si})=t^2$ for $2t\mid p^m+1$, $1\leq i< \frac{p^m+1}{t}$ and $\alpha^{si}\neq-1$.
\item[(vii)] if $2t\nmid p^m+1$, then $\delta(\alpha^{si}-\alpha^{sj})
    =t^2$, where $1\leq i,j< \frac{p^m+1}{t}$ and $i\neq j$.
\item[(vii)] if $2t\mid p^m+1$, then $\delta(\alpha^{si}-\alpha^{sj})
    =2t^2$, where $1\leq i,j< \frac{p^m+1}{t}$, $\alpha^{si},\alpha^{sj}\neq-1$, $i\neq j$ and $\alpha^{si}-\alpha^{sj}\neq\pm1$.
\item[(ix)] $\delta(b)=0$ for $b\neq \pm1,\alpha^{si}-\alpha^{sj}$, where $0\leq i,j< \frac{p^m+1}{t}$.\hfill$\square$
\end{itemize}
\end{lemma}
By analyzing the results of Lemma \ref{lemma ds s f4}, we have the following results.
\begin{theorem}
Let $f(x)=x^d$ be a power mapping defined over $\mathbb{F}_{p^n}$, where $d=s(p^m-1)$, $n=2m$ and $\frac{p^m+1}{t}>3$. For any $b\in\mathbb{F}_{p^n}$, the differential spectrum of $f$ is shown as following.

(i) For case $2t\nmid p^m+1$, Table \ref{ds41table} holds.

(ii) For case $2t\mid p^m+1$, $6t\nmid p^m+1$, Table \ref{ds42table} holds.

(iii) For case $6t\mid p^m+1$, Table \ref{ds43table} holds.
\begin{table}[H]
\centering
 \caption{$2t\nmid p^m+1$}
\label{ds41table}
\begin{tabular}{l l}
\hline
$\delta(1,b)$&  Frequency \\
[0.5ex]
\hline
0 & $\frac{t^2p^n-(p^m+2-t)p^m-3t^2+t-1}{t^2}$ \\

$t^2$ & $\frac{p^{2m}+(2-3t)p^m+2t^2-3t+1}{t^2}$\\

$t^2-t$  &
  $\frac{2(p^m-t+1)}{t}$  \\

1 & 2\\
 
$t(p^m-1)-1$  & 1\\
\hline 
\end{tabular}
\end{table}
\begin{table}[H]
\centering
 \caption{$2t\mid p^m+1$, $6t\nmid p^m+1$}
\label{ds42table}
\begin{tabular}{l l}
\hline
$\delta(1,b)$&  Frequency \\
[0.5ex]
\hline
0 & $\frac{2t^2p^n-(p^m+2)p^{m}-6t^2-1}{2t^2}$ \\

$2t^2$  &
$\frac{p^{2m}+(2-6t)p^m+8t^2-6t+1}{2t^2}$ \\

$2t^2-t$ & $\frac{2(p^m-2t+1)}{t}$ \\

$t^2$  &
$\frac{p^m-2t+1}{t}$\\

$t^2-t$ & 2\\

1 & 2 \\

$t(p^m-1)-1$  &1 \\
\hline
\end{tabular}
\end{table}
\begin{table}[H]
\centering
 \caption{$6t\mid p^m+1$}
\label{ds43table}
\begin{tabular}{l l}
\hline
$\delta(1,b)$&  Frequency \\
[0.5ex]
\hline
$0$ & $\frac{2t^2p^n-(p^m+2)p^{m}-2t^2-1}{2t^2}$ \\

$2t^2$ &
  $\frac{p^{2m}+(2-6t)p^m+4t^2-6t+1}{2t^2}$ \\

$2t^2-t$   & $\frac{2(p^m-2t+1)}{t}$ \\

$t^2$ & $\frac{p^m-2t+1}{t}$ \\

$2t^2+1$ &  2\\

$t^2-t$ & 2\\

$t(p^m-1)-1$  & 1\\
\hline
\end{tabular}
\end{table}
\hfill$\square$
\end{theorem}

To empirically verify our theoretical findings, we conduct the following numerical experiments.
\begin{example} \rm{
Let $f(x)=x^d=x^{5^2-1}$ defined over $\mathbb{F}_{5^{4}}$. Then the differential spectrum of $f=x^d$ is given by
\begin{equation*}
 \begin{aligned}
  \mathbb{DS}=
  \bigg\{\omega_0=286,\omega_{1}=74,\omega_{2}=264,\omega_{23}=1 \bigg\}.
   \end{aligned}
\end{equation*}}
\end{example}

\begin{example} \rm{
Let $f(x)=x^d=x^{2(11^2-1)}$ defined over $\mathbb{F}_{11^{4}}$. Then the differential spectrum of $f=x^d$ is given by
\begin{equation*}
 \begin{aligned}
  \mathbb{DS}=
  \bigg\{\omega_0=10978,\omega_{1}=2,\omega_{2}=120,\omega_{4}=3540 ,\omega_{239}=1\bigg\}.
   \end{aligned}
\end{equation*} }
\end{example}

\section{Boomerang spectrum of $f(x)=x^{s(p^m-1)}$ over $\mathbb{F}_{p^n}$}\label{luo the bs}
In \cite{kpm-1}, Hu et al. determined the boomerang spectrum of $x^{s(p^m-1)}$ based on the differential spectrum of $x^{s(p^m-1)}$, where $n=2m$ and $\gcd(s,p^m+1)=1$. So what happens if $\gcd(s,p^m+1)\geq1$? Let $f(x)=x^d$, where $d=s(p^m-1)$, $n=2m$ and $\gcd(s,p^m+1)=t$. In this section,  by studying the number of rational points on some curves, and without relying on the differential spectrum of $f$, we determine the boomerang spectrum of $f$. For any fixed $b\in \mathbb{F}_{p^n}^*$, denote by $\beta(b)$ the number of solutions to 
\begin{eqnarray}\label{boom eq f2}
    \left \{
\begin{array}{lll}
x^d-y^d=b,\\
(x+1)^{d}-(y+1)^{d}=b,\\
\end{array}\right.
\end{eqnarray}
in $\mathbb{F}_{p^n}^2$.

\subsection{$p=2$}
In this subsection, we determine the boomerang spectrum of $f(x)=x^d$ over $\mathbb{F}_{2^n}$.
\begin{lemma}\label{bs 2 lemma}
Let $\alpha=\psi^{2^m-1}$ and $\frac{2^m+1}{t}>3$. With the notations above, we get 
\begin{itemize}
\item[(i)]  $\beta(1)=2$ for $3t\nmid2^m+1$. $\beta(1)=4t^4-4t^3+2t^2+2$ for $3t\mid2^m+1$. 
\item[(ii)] $\beta(1+\alpha^{si})=2t(t-1)(2^m+2t^2-4t)$ for $1\leq i<\frac{2^m+1}{t}$, $1+\alpha^{si}\neq\alpha^{\frac{2^m+1}{3}}, \alpha^{\frac{2(2^m+1)}{3}}$.
\item[(iii)] $\beta(\alpha^{si}+\alpha^{sj})=2t^4+2t^2(t-1)^2=4t^4-4t^3+2t^2$ for $1\leq i,j<\frac{2^m+1}{t}$, $i\neq j$, $\alpha^{si}+\alpha^{sj}\neq1$.
\item[(iv)] $\beta(\alpha^{\frac{2^m+1}{3}})=\beta(\alpha^{\frac{2(2^m+1)}{3}})=0$ for $3t\nmid2^m+1$. $\beta(\alpha^{\frac{2^m+1}{3}})
    =\beta(\alpha^{\frac{2(2^m+1)}{3}})=2t(t-1)(2^m+2t^2-4t)+4t^2$ for $3t\mid2^m+1$. 
\item[(v)] $\beta(b)=0$ for $b\neq0, 1, \alpha^{si}+\alpha^{sj}$, where $0\leq i,j<\frac{2^m+1}{t}$.
\end{itemize}
\end{lemma}
\begin{proof}
(i) When $b=1$ and $3t\nmid2^m+1$, (\ref{boom eq f2}) have only two solutions $(x,y)=(0,1), (1,0)$ in
$\mathbb{F}_{2^n}\times\mathbb{F}_{2^n}$.

If $b=1$ and $3t\mid2^m+1$, we let 
\begin{equation*}
  \frac{ls(2^n-1)}{2^m+1}\equiv \frac{2^n-1}{3}\,\,({\rm mod}\,\,2^n-1),
\end{equation*}
where $1\leq l<\frac{2^m+1}{t}$. Record $C_{k_1,k_2}$ in (\ref{c12}). By Lemma \ref{luo wang}, when $(x,y)\in C_{k_1,k_2}\times C_{k_3,k_4}$, we have the following results.
\begin{table}[H]
\centering
 \caption{Number of solutions to (\ref{boom eq f2}) in 
 $C_{k_1,k_2}\times C_{k_3,k_4}$ when $b=1$}
\label{k1234table}
\begin{tabular}{cc}
%\hline
\hline
%inserts double horizontal lines
$(k_1,k_2)\times(k_3,k_4)$&
Number of solutions to (\ref{boom eq f2}) when $b=1$ (total) \\
[0.5ex]
\hline
 \tabincell{c}{$(2l+\frac{(2^m+1)j_1}{t},
  2l+\frac{(2^m+1)j_2}{t})$\\
  $\times
  (l+\frac{(2^m+1)j_3}{t},l+\frac{(2^m+1)j_4}{t})$
  \\$0\leq j_i<t$, $0\leq i\leq4$} & $t^4-2t^3+t^2$ \\
\hline
\tabincell{c}{$(l+\frac{(2^m+1)j_1}{t},
l+\frac{(2^m+1)j_2}{t})$\\
$\times(2l+\frac{(2^m+1)j_3}{t},
  2l+\frac{(2^m+1)j_4}{t})$\\
  $0\leq j_i<t$, $0\leq i\leq4$} & $t^4-2t^3+t^2$ \\
\hline
\tabincell{c}{$(2l+\frac{(2^m+1)j_1}{t},
  l+\frac{(2^m+1)j_2}{t})$\\
  $\times
  (l+\frac{(2^m+1)j_3}{t},
  2l+\frac{(2^m+1)j_4}{t})$\\
  $0\leq j_i<t$, $0\leq i\leq4$} & $t^4$ \\
\hline
\tabincell{c}{$(l+\frac{(2^m+1)j_1}{t},
  2l+\frac{(2^m+1)j_2}{t})$\\
  $\times
  (2l+\frac{(2^m+1)j_3}{t},
  l+\frac{(2^m+1)j_4}{t})$\\
  $0\leq j_i<t$, $0\leq i\leq4$} & $t^4$ \\
\hline
Other Cases &  $0$ \\
\hline
\end{tabular}
\end{table}
Note that (\ref{boom eq f2}) also have two solutions $(x,y)=(0,1), (1,0)$ in
$\mathbb{F}_{2^n}\times\mathbb{F}_{2^n}$. So, $\beta(1)=4t^4-4t^3+2t^2+2$ for $3t\mid2^m+1$. 

(ii)-(iv) Let $\frac{2^m+1}{t}=h$. When $b=1+\alpha^s$, (\ref{boom eq f2}) have solutions only in 
\begin{equation*}
  \bigcup\limits_{\substack {0\leq j_i<t,\\0\leq i\leq4}}\bigg(C_{1+hj_1,hj_2}
  \times C_{hj_3,1+hj_4}\bigcup
  C_{hj_1,1+hj_2}
  \times C_{1+hj_3,hj_4}\bigg),
\end{equation*}
\begin{equation*}
  \bigcup\limits_{\substack {0\leq j_i<t,\\0\leq i\leq4}}\bigg(C_{hj_1,hj_2}
  \times C_{1+hj_3,1+hj_4}\bigcup
  C_{1+hj_1,1+hj_2}
  \times C_{hj_3,hj_4}\bigg).
\end{equation*}
When $b=\alpha^s+\alpha^{2s}$, (\ref{boom eq f2}) have solutions only in 
\begin{equation*}
  \bigcup\limits_{\substack {0\leq j_i<t,\\0\leq i\leq4}}\bigg(C_{2+hj_1,1+hj_2}
  \times C_{1+hj_3,2+hj_4}\bigcup
  C_{1+hj_1,2+hj_2}
  \times C_{2+hj_3,1+hj_4}\bigg),
\end{equation*}
\begin{equation*}
  \bigcup\limits_{\substack {0\leq j_i<t,\\0\leq i\leq4}}\bigg(C_{1+hj_1,1+hj_2}
  \times C_{2+hj_3,2+hj_4}\bigcup
  C_{2+hj_1,2+hj_2}
  \times C_{1+hj_3,1+hj_4}\bigg).
\end{equation*}
When $b=\alpha^{\frac{2^m+1}{3}}$ and $3t\nmid2^m+1$, (\ref{boom eq f2}) have no solutions in
$\mathbb{F}_{2^n}\times\mathbb{F}_{2^n}$.
When $b=\alpha^{\frac{2^m+1}{3}}$ and $3t\mid2^m+1$, (\ref{boom eq f2}) have solutions only in 
\begin{equation*}
  \bigcup\limits_{\substack {0\leq j_i<t,\\0\leq i\leq4}}\bigg(C_{hj_1,hj_2}
  \times
  C_{l+hj_3,l+hj_4}
\bigcup C_{l+hj_1,l+hj_2}
  \times
  C_{hj_3,hj_4}
  \bigg),
\end{equation*}
\begin{equation*}
  \bigcup\limits_{\substack {0\leq j_i<t,\\0\leq i\leq4}}\bigg(C_{l+hj_1,hj_2}
  \times
  C_{hj_3,l+hj_4}
\bigcup C_{hj_1,l+hj_2}
  \times
  C_{l+hj_3,hj_4}
  \bigg),
\end{equation*}
\begin{equation*}
  \bigcup\limits_{\substack {0\leq j_i<t,\\0\leq i\leq4}}\bigg(C_{2l+hj_1,l+hj_2}
  \times \{0\}\bigcup 
  \{0\}\times C_{2l+hj_3,l+hj_4}\bigg),
\end{equation*}
\begin{equation*}
  \bigcup\limits_{\substack {0\leq j_i<t,\\0\leq i\leq4}}\bigg(C_{l+hj_1,2l+hj_2}
  \times \{1\}\bigcup 
  \{1\}\times C_{l+hj_3,2l+hj_4}\bigg),
\end{equation*}
where $\frac{ls(2^n-1)}{2^m+1}\equiv \frac{2^n-1}{3}\,\,({\rm mod}\,\,2^n-1)$.
By Lemma \ref{luo wang}, we obtain the result.

(v) The desired conclusion is an immediate consequence of combining results (i)-(iv) with Lemma \ref{doubesame}.
\end{proof}

By Lemma \ref{bs 2 lemma}, we determine the boomerang spectrum of $f$ in the following result.
\begin{theorem}
Let $f(x)=x^d$ defined over $\mathbb{F}_{2^n}$, where $d=s(2^m-1)$, $n=2m$ and $\gcd(s,2^m+1)=t$. For any $b\in\mathbb{F}_{2^n}^*$, the boomerang spectrum of $f$ is shown as following.

(i) For case $3t\nmid 2^m+1$, Table \ref{bs21table} holds.

(ii) For case $3t\mid 2^m+1$, Table \ref{bs22table} holds.
\begin{table}[H]
\centering
 \caption{$3t\nmid2^m+1$}
\label{bs21table}
\begin{tabular}{l l}
%\hline
\hline
%inserts double horizontal lines
$\beta_f(1,b)$   &  Frequency \\
[0.5ex]
\hline
$0$ & $\frac{(2t^2-1)2^n-(2-t)2^m-4t^2+t-1}{2t^2}$ \\

2 & 1\\

$4t^4-4t^3+2t^2$   & $\frac{2^n+(2-3t)2^m+2t^2-3t+1}{2t^2}$ \\

$2t(t-1)(2^m+2t^2-4t)$ & $\frac{2^m-t+1}{t}$ \\
\hline
\end{tabular}
\end{table}
\begin{table}[H]
\centering
 \caption{$3t\mid2^m+1$}
\label{bs22table}
\begin{tabular}{l l}
%\hline
\hline
%inserts double horizontal lines
$\beta_f(1,b)$   &  Frequency  \\
[0.5ex]
\hline
$0$ & $\frac{(2t^2-1)2^n-(2-t)2^m-6t^2+t-1}{2t^2}$ \\

$4t^4-4t^3+2t^2+2$ & 1\\

$4t^4-4t^3+2t^2$   & $\frac{2^n+(2-3t)2^m-3t+1}{2t^2}$ \\

$2t(t-1)(2^m+2t^2-4t)$ & $\frac{2^m-t+1}{t}$ \\

$2t(t-1)(2^m+2t^2-4t)+4t^2$ & 2\\
\hline
\end{tabular}
\end{table}
\end{theorem}

\subsection{$p=3$}
In this subsection, we characterize the boomerang spectrum of $f(x)=x^d$ over $\mathbb{F}_{3^n}$, employing techniques analogous to those established in Lemma \ref{bs 2 lemma}, with similar derivations omitted.

\begin{lemma}\label{bs 3 lemma}
Let $\alpha=\psi^{3^m-1}$ and $\frac{3^m+1}{t}>3$. Employing the notation introduced earlier, we demonstrate that:
\begin{itemize}
\item[(i)] $\beta(1)=\beta(2)=0$
     for $2t\nmid3^m+1$. $\beta(1)=\beta(2)=t(t-1)(3^m+t^2-3t+2)$ for $2t\mid3^m+1$.
\item[(ii)] if $2t\nmid3^m+1$, then $\beta(1-\alpha^{si})
    =\beta(\alpha^{si}-1)
    =t(t-1)(3^m+t^2-3t)$, where $1\leq i< \frac{3^m+1}{t}$.
\item[(iii)] if $2t\mid3^m+1$, then $\beta(1-\alpha^{si})
    =\beta(\alpha^{si}-1)
    =t(t-1)(3^m+4t^2-4t)$, where $1\leq i< \frac{3^m+1}{t}$, $\alpha^{si}\neq2$.
\item[(iv)] if $2t\nmid3^m+1$, then $\beta(2\alpha^{si})
    =0$, where $1\leq i< \frac{3^m+1}{t}$.
\item[(v)] if $2t\mid3^m+1$, then $\beta(2\alpha^{si})
    =t^4-2t^3+t^2$, where $1\leq i< \frac{3^m+1}{t}$, $\alpha^{si}\neq2$.
\item[(vi)] if $2t\nmid3^m+1$, then $\beta(\alpha^{si}-\alpha^{sj})
    =t^4-2t^3+t^2$, where $1\leq i,j< \frac{3^m+1}{t}$, $i\neq j$.
\item[(vii)] if $2t\mid3^m+1$, then $\beta(\alpha^{si}-\alpha^{sj})
    =4t^4-4t^3+2t^2$, where $1\leq i,j< \frac{3^m+1}{t}$, $\alpha^{si},\alpha^{sj}\neq2$, $i\neq j$.
\item[(viii)] $\beta(b)=0$ for $b\neq0, \pm1, \alpha^{si}-\alpha^{sj}$, where $0\leq i,j< \frac{3^m+1}{t}$. \hfill$\square$
\end{itemize}
\end{lemma}

Building upon Lemma \ref{bs 3 lemma}, we derive the following results.
\begin{theorem}
Let $f(x)=x^d$ be a power mapping defined over $\mathbb{F}_{3^n}$, where $d=s(3^m-1)$, $n=2m$, $\gcd(s,3^m+1)=t$ and $\frac{3^m+1}{t}>3$. For all $b\in\mathbb{F}_{3^n}^*$, the boomerang spectrum of $f$ is shown as following.

(i) For case $2t\nmid 3^m+1$, Table \ref{bs31table} holds.

(ii) For case $2t\mid 3^m+1$, Table \ref{bs32table} holds.
\begin{table}[H]
\centering
 \caption{$2t\nmid3^m+1$}
\label{bs31table}
\begin{tabular}{l l}
%\hline
\hline
%inserts double horizontal lines
$\beta_f(1,b)$   &  Frequency  \\
[0.5ex]
\hline
$0$ & $\frac{(t^2-1)3^n-(2-t)3^m-t^2+t-1}{t^2}$ \\

$t^4-2t^3+t^2$   & $\frac{3^n+(2-3t)3^m+2t^2-3t+1}{t^2}$ \\

$t(t-1)(3^m+t^2-3t)$ & $\frac{2(3^m-t+1)}{t}$ \\
\hline
\end{tabular}
\end{table}
\begin{table}[H]
\centering
 \caption{$2t\mid3^m+1$}
\label{bs32table}
\begin{tabular}{l l}
%\hline
\hline
%inserts double horizontal lines
$\beta_f(1,b)$   &  Frequency  \\
[0.5ex]
\hline
$0$ & $\frac{2t^23^n-(3^m+2)3^{m}-2t^2-1}{2t^2}$ \\

$t^4-2t^3+t^2$   & $\frac{3^m-2t+1}{t}$ \\

$4t^4-4t^3+2t^2$ & $\frac{3^{n}+(2-6t)3^m+8t^2-6t+1}{2t^2}$ \\

$t(t-1)(3^m+4t^2-4t)$ & $\frac{2(3^m-2t+1)}{t}$ \\

$t(t-1)(3^m+t^2-3t+2)$ & $2$ \\
\hline
\end{tabular}
\end{table}
 \hfill$\square$
\end{theorem}

\subsection{$p>3$}
In this subsection, we compute the boomerang spectrum of $f(x)=x^d$ over $\mathbb{F}_{p^n}$ ($p>3$) using methods parallel to Lemma \ref{bs 2 lemma}, and thus we skip similar proof steps.
\begin{lemma}\label{bs 4 lemma}
Let $\alpha=\psi^{p^m-1}$ and $\frac{p^m+1}{t}>3$. With the notations above, we have 
\begin{itemize}
\item[(i)] $\beta(1)=\beta(p-1)=0$
     for $6t\nmid p^m+1$. $\beta(1)=\beta(p-1)=4t^4-4t^3+2t^2$ for $6t\mid p^m+1$.
\item[(ii)] $\beta(\pm\alpha^{\frac{2^m+1}{3}})=\beta(\pm\alpha^{\frac{2(2^m+1)}{3}})=0$ for $6t\nmid p^m+1$. $\beta(\pm\alpha^{\frac{2^m+1}{3}})
    =\beta(\pm\alpha^{\frac{2(2^m+1)}{3}})=t(t-1)(p^m+4t^2-4t)+2t^2$ for $6t\mid p^m+1$. 
\item[(iii)] $\beta(2)=\beta(p-2)=0$
     for $2t\nmid p^m+1$. $\beta(2)=\beta(p-2)=t(t-1)(p^m+t^2-3t)$ for $2t\mid p^m+1$.
\item[(iv)] if $2t\nmid p^m+1$, then $\beta(1-\alpha^{si})
    =\beta(\alpha^{si}-1)
    =t(t-1)(p^m+t^2-3t)$, where $1\leq i< \frac{p^m+1}{t}$.
\item[(v)] if $2t\mid p^m+1$, then $\beta(1-\alpha^{si})
    =\beta(\alpha^{si}-1)
    =t(t-1)(p^m+4t^2-4t)$, where $1\leq i< \frac{p^m+1}{t}$, $\alpha^{si}\neq-1$, $1-\alpha^{si}\neq-\alpha^{\frac{p^m+1}{3}},-\alpha^{\frac{2(p^m+1)}{3}}$,
    $\alpha^{si}-1\neq\alpha^{\frac{p^m+1}{3}},\alpha^{\frac{2(p^m+1)}{3}}$.
\item[(vi)] $\beta(2\alpha^{si})
    =0$ for $1\leq i< \frac{p^m+1}{t}$, $2t\nmid p^m+1$. $\beta(2\alpha^{si})
    =t^4-2t^3+t^2$ for $1\leq i< \frac{p^m+1}{t}$, $\alpha^{si}\neq-1$, $2t\mid p^m+1$.
\item[(vii)] if $2t\nmid p^m+1$, then $\beta(\alpha^{si}-\alpha^{sj})
    =t^4-2t^3+t^2$, where $1\leq i,j< \frac{p^m+1}{t}$, $i\neq j$.
\item[(viii)] if $2t\mid p^m+1$, then $\beta(\alpha^{si}-\alpha^{sj})
    =4t^4-4t^3+2t^2$, where $1\leq i,j< \frac{p^m+1}{t}$, $\alpha^{si},\alpha^{sj}\neq-1$, $i\neq j$, $\alpha^{si}-\alpha^{sj}\neq\pm1$.
\item[(ix)] $\beta(b)=0$ for $b\neq0,\pm1, \alpha^{si}-\alpha^{sj}$, where $0\leq i,j< \frac{p^m+1}{t}$. \hfill$\square$
\end{itemize}
\end{lemma}

Following the framework established in Lemma \ref{bs 4 lemma}, we obtain the following.
\begin{theorem}
Let $f(x)=x^d$ defined over $\mathbb{F}_{p^n}$, where $d=s(p^m-1)$, $n=2m$, $\gcd(s,p^m+1)=t$, $p>3$ and $\frac{p^m+1}{t}>3$. For any $b\in\mathbb{F}_{p^n}^*$, the boomerang spectrum of $f$ is shown as following.

(i) For case $2t\nmid p^m+1$, Table \ref{bs41table} holds.

(ii) For case $2t\mid p^m+1$, $6t\nmid p^m+1$, Table \ref{bs42table} holds.

(iii) For case $6t\mid p^m+1$, Table \ref{bs43table} holds.
\begin{table}[H]
\centering
 \caption{$2t\nmid p^m+1$}
\label{bs41table}
\begin{tabular}{l l}
%\hline
\hline
%inserts double horizontal lines
$\beta_f(1,b)$   &  Frequency  \\
[0.5ex]
\hline
$0$ & $\frac{(t^2-1)p^n-(2-t)p^m-t^2+t-1}{t^2}$ \\

$t^4-2t^3+t^2$   & $\frac{p^n+(2-3t)p^m+2t^2-3t+1}{t^2}$ \\

$t(t-1)(p^m+t^2-3t)$ & $\frac{2(p^m-t+1)}{t}$ \\
\hline
\end{tabular}
\end{table}
\begin{table}[H]
\centering
 \caption{$2t\mid p^m+1$, $6t\nmid p^m+1$}
\label{bs42table}
\begin{tabular}{l l}
%\hline
\hline
%inserts double horizontal lines
$\beta_f(1,b)$   &  Frequency  \\
[0.5ex]
\hline
$0$ & $\frac{2t^2p^n-(p^m+2)p^{m}-2t^2-1}{2t^2}$ \\

$t^4-2t^3+t^2$   & $\frac{p^m-2t+1}{t}$ \\

$4t^4-4t^3+2t^2$ & $\frac{p^{n}+(2-6t)p^m+8t^2-6t+1}{2t^2}$ \\

$t(t-1)(p^m+4t^2-4t)$ & $\frac{2(p^m-2t+1)}{t}$ \\

$t(t-1)(p^m+t^2-3t)$ & $2$ \\
\hline
\end{tabular}
\end{table}
\begin{table}[H]
\centering
 \caption{$6t\mid p^m+1$}
\label{bs43table}
\begin{tabular}{l l}
%\hline
\hline
%inserts double horizontal lines
$\beta_f(1,b)$   &  Frequency  \\
[0.5ex]
\hline
$0$ & $\frac{2t^2p^n-(p^m+2)p^{m}-2t^2-1}{2t^2}$ \\

$t^4-2t^3+t^2$   & $\frac{p^m-2t+1}{t}$ \\

$4t^4-4t^3+2t^2$ & $\frac{p^{n}+(2-6t)p^m+8t^2-6t+1}{2t^2}$ \\

$t(t-1)(p^m+4t^2-4t)$ & $\frac{2(p^m-4t+1)}{t}$ \\

$t(t-1)(p^m+4t^2-4t)+2t^2$ & $4$ \\

$t(t-1)(p^m+t^2-3t)$ & 2\\
\hline
\end{tabular}
\end{table}
  \hfill$\square$
\end{theorem}

We present examples confirming our findings.
\begin{example}\rm{
Let $f(x)=x^d=x^{3(3^4-1)}$ defined over $\mathbb{F}_{3^{8}}$. Then the boomerang spectrum of $f=x^d$ is given by
\begin{equation*}
 \begin{aligned}
  \mathbb{BS}=
  \bigg\{\nu_0=3440,\nu_{2}=3120\bigg\}.
   \end{aligned}
\end{equation*} }
\end{example}
\begin{example}\rm{
Let $f(x)=x^d=x^{2(7^2-1)}$ defined over $\mathbb{F}_{7^{4}}$. Then the boomerang spectrum of $f=x^d$ is given by
\begin{equation*}
 \begin{aligned}
  \mathbb{BS}=
  \bigg\{\nu_0=1800,\nu_{4}=552,\nu_{94}=48\bigg\}.
   \end{aligned}
\end{equation*}}
\end{example}

\section{Conclusion}\label{conclusion}
In this paper, we completely determined the differential spectrum and boomerang spectrum of the power mapping $f(x)=x^{s(p^m-1)}$ over $\mathbb{F}_{p^n}$ with $n=2m$ and $\gcd(s,p^m+1)\geq1$. From the results, our findings on the differential spectrum and boomerang spectrum of $f$ extended the results in \cite{kpm-1, bu2} from $t=1$ to general case. Methodologically, we determined the differential spectrum and boomerang spectrum of $f$ by analyzing the number of rational points on some curves, which is different from the method in \cite{kpm-1, bu2}.

\end{document}